\documentclass[10pt]{article}
\usepackage{subcaption} %
\usepackage{graphicx} % for pdf, bitmapped graphics files
\usepackage{times} % assumes new font selection scheme installed
\usepackage{amsmath} % assumes amsmath package installed
\usepackage{amssymb,amsthm} % assumes amsmath package installed
\usepackage[colorlinks=true]{hyperref}
\usepackage{relsize} % use \mathlarger{}

\newtheorem{lemma}{Lemma}

\newcommand{\bH}{\widehat{H}}

\newcommand{\bJ}{\widehat{J}}

\newcommand{\bL}{\widehat{L}}

\newcommand{\bS}{\widehat{S}}

\newcommand{\bW}{\widehat{W}}
\newcommand{\bX}{\widehat{X}}

\newcommand{\cD}{\mathcal{D}}

\newcommand{\cL}{\mathcal{L}}
\newcommand{\cK}{\mathcal{K}}
\newcommand{\ocK}{\overline{\mathcal{K}}_0}

\newcommand{\ocR}{\overline{\mathcal{R}}_0}

\newcommand{\cH}{\mathcal{H}}

\newcommand{\dotex}{\frac{d}{dt}}

\newcommand{\Tr}[1]{\rm{Tr}\left(#1\right)}

\title{\LARGE \bf Heisenberg formulation of adiabatic elimination
for open quantum systems with two time-scales}

\author{Fran\c{c}ois-Marie Le R\'{e}gent \thanks{
Alice\&Bob,
53 boulevard du G\'{e}n\'{e}ral Martial Valin, 75015 Paris. \texttt{francois-marie.le-regent@alice-bob.com}} $~^{\dag}$ \quad Pierre Rouchon \thanks{Laboratoire de Physique de l’Ecole normale sup\'{e}rieure, Mines Paris-PSL, Inria, ENS-PSL, Universit\'{e} PSL, CNRS, Paris. \texttt{pierre.rouchon@minesparis.psl.eu}}
}

\begin{document}

\maketitle
\thispagestyle{empty}
\pagestyle{empty}

%%%%%%%%%%%%%%%%%%%%%%%%%%%%%%%%%%%%%%%%%%%%%%%%%%%%%%%%%%%%%%%%%%%%%%%%%%%%%%%%
\begin{abstract}

Consider an open quantum system governed by a Gorini–Kossakowski\\–Sudarshan–Lindblad  master equation with two times-scales: a fast one, exponentially converging towards a linear subspace of quasi-equilibria;  a slow one resulting small decoherence and Hamiltonian dynamics. Usually adiabatic elimination is performed in the Schr\"{o}dinger picture. We propose here an Heisenberg formulation where the invariant operators attached to the fast decay dynamics towards the quasi-equilibria subspace  play a key role. Based on geometric singular perturbations, asymptotic expansions of the Heisenberg slow dynamics and of the fast invariant linear subspaces are proposed. They exploit Carr's approximation lemma from center-manifold and bifurcation theory. Second-order expansions are detailed and shown to ensure preservation, up to second-order terms, of the  complete positivity for the slow  propagator   on a slow time-scale. Such expansions can be exploited numerically to derive reduced-order dynamical  models.

\end{abstract}

%%%%%%%%%%%%%%%%%%%%%%%%%%%%%%%%%%%%%%%%%%%%%%%%%%%%%%%%%%%%%%%%%%%%%%%%%%%%%%%%
\section{Introduction}

In the quantum physics community, adiabatic elimination is widely used to analyze the dynamics of open and dissipative quantum systems (see e.g. \cite{BrionJPA07,ZanardiPRA2015,MacieszczakGutaLesanovskyEtAl2016,AzouitQST2017,BurgarthQ2019}. It corresponds in fact to a perturbation techniques known in dynamical and control system theory as singular perturbations for slow/fast systems. It is related to the Tikhonov approximation theorem (see, e.g., ~\cite{VerhulstBook2005,kokotovic-book-1}) and its coordinate-free formulation due to Fenichel~\cite{Fenichel79}. The notion of invariant slow manifolds plays a crucial role  for  dynamical systems with two time-scales dynamics: the fast and exponentially converging ones and the slow ones of reduced dimension. In this context, adiabatic elimination produces low dimensional dynamical models via the derivation from the original slow/fast differential equations of the slow differential equations governing the evolution on the invariant slow manifold. For open quantum system governed by the deterministic Gorini–Kossakowski–Sudarshan–Lindblad (GKSL) master equation, adiabatic elimination is usually performed in the Schr\"{o}dinger picture.

We perform here adiabatic elimination in the Heisenberg picture where the invariant operators associated to the fast dynamics play a crutial role: they are used to describe the slow dynamics but also to define the equations characterizing the fast invariant linear subspace. As far a we know, such Heisenberg point of view has not been considered in such a systematic and general way, despite the fact that, for stochastic quantum systems, the Heisenberg stochastic evolutions play a central role (see, e.g.,~\cite{BouteS2008,GoughJMP2010}). In particular, our derivation relies on very general two time-scale assumptions: we only assume an exponentially fast convergence towards a linear subspace of quasi-equilibria those structure does not necessarily correspond to a decohence free subspace; the slow dynamics can result from arbitrary perturbations either Hamiltonian or Lindbladian; we do not assume a tensor-product structure where the fast decay is due to local decoherence fast dynamics of some  sub-systems.

Combining two asymptotic expansions, a first one for the slow dynamics and a complementary one providing the  set of linear equations characterizing the exponentially fast decaying sub-space, we show how to approximate up to exponentially small corrections the propagator over a slow time-scale (see lemma~\ref{lem:Approx}). We explain how to compute the order $n$ corrections knowing the correction of order $r<n$. The second-order approximation of the slow dynamics is shown to preserve complete positivity in the following sense:  its second-order propagator over a slow time-scale corresponds, up to second-order correction, to a Trace Preserving and Completely Positive (TPCP) map (see lemma~\ref{lem:approx2}). Such preservation has been shown in specific cases for the second-order as in~\cite{MacieszczakGutaLesanovskyEtAl2016,AzouitQST2017} or for  first-order as in~\cite{ZanarC2014PRL,BurgarthPRA21}.

In section~\ref{sec:slowfastdyn}, the slow/fast structure of the GKLS differential equations is detailed either in Schr\"{o}dinger picture with an orthonormal basis $(\bS_d)_{1\leq d \leq \bar d}$ for  quasi-equilibria quantum states  but also in the Heisenberg picture with the associated basis of quasi-invariant operators $(\bJ_d)_{1\leq d \leq \bar d}$.
In section~\ref{sec:slowManifold}, we detail the asymptotic expansion of the slow dynamics with an Heisenberg point of view.
In section~\ref{sec:fastMani}, the set of independent linear equations describing the fast invariant subspace is constructed and its approximation at any order is given.
Section~\ref{sec:slowPropa} combines lemma~\ref{lem:slowdyn} of section~\ref{sec:slowManifold} and lemma~\ref{lem:fastInv} of section~\ref{sec:fastMani} to prove lemmas~\ref{lem:Approx} and~\ref{lem:approx2}, the approximate TPCP character  of slow propagators over a slow time-scale.

Throughout this paper, the underlying Hilbert space $\cH$ is assumed to be of finite dimension. This ensure uniqueness, existence and convergence of these asymptotic expansions versus the small parameter $\epsilon$. The calculations below use the language of operators. Thus they can be used, at least formally, even for an infinite dimensional Hilbert space despite the fact that precise mathematical justifications relying on functional analysis methods   are not straightforward.

%%%%%%%%%%%%%%%%%%%%%%%%%%%%%%%%%%%%%%%%%%%%%%%%%%%%%%%%%%%%%%%%%%%%%%%%%%%%%%%%
\section{Slow/fast dynamics} \label{sec:slowfastdyn}

\subsection{Singular perturbations for finite dimensional, linear and time-invariant systems}

 Take a linear time-invariant system of finite dimension
 $$
 \dotex \xi = (A_0 + \epsilon A_1)\xi
 $$
 where $\xi$ is a real vector of finite dimension $\bar D$, $A_0$ and $A_1$ are $\bar D\times \bar D$ matrices with real entries and $\epsilon$ is a small parameter. Assuming a  slow/fast structure means that $A_0$ can be block diagonalized in two blocks:
  $$
  A_0= P_0 \left(
       \begin{array}{cc}
        0 & 0 \\
        0 & \Gamma_0 \\
       \end{array}
      \right)
  P_0^{-1}
  $$
  where $P_0$ is invertible and $\Gamma_0$ is an Hurwitz (stable) matrix of dimension $\bar D - \bar d>0$, the dimension of the fast dynamics and  where $\bar d>0 $ is the dimension of the slow dynamics.
  Standard perturbation theory (see~\cite{kato-book-66}) ensures that, for $\epsilon$ small enough, one has a similar block decomposition:
$$
  A_0+ \epsilon A_1= P(\epsilon) \left(
       \begin{array}{cc}
       \Delta(\epsilon) & 0 \\
        0 & \Gamma(\epsilon) \\
       \end{array}
      \right)
  P^{-1}(\epsilon)
$$
where the matrices $P(\epsilon)$, $ \Delta(\epsilon) $ and $\Gamma(\epsilon)$ are analytic versus $\epsilon$ with
$P(0)=P_0$, $\Delta(0)=0$ and $\Gamma(0)=\Gamma_0$. Geometrically, exists, for $\epsilon$ small enough, two invariant linear subspaces:
\begin{itemize}
 \item the slow one, of dimension $\bar d$, corresponding to the slow evolution governed by the propagator $e^{t\Delta(\epsilon)}$
 \item the fast one, of dimension $\bar D - \bar d$, corresponding to the fast and exponentially stable evolution governed by the propagator
 $e^{t\Gamma(\epsilon)}$.
\end{itemize}

\subsection{Slow/fast GKSL quantum dynamics}
All the developments below combine the above  dynamics  structure  with non commutative computations with operators  used to describe the decoherence dynamics of open-quantum systems.

  Consider the time-varying density operator $\rho_t$ on underlying Hilbert space $\cH$ of finite dimension obeying to the following dynamics
\begin{equation}\label{eq:dynL0L1}
 \dotex \rho_t = \cL_0(\rho_t) + \epsilon \cL_1(\rho_t)
\end{equation}
 where $\epsilon $ is a small positive parameter and where the GKSL linear super-operators $\cL_0$ and $\cL_1$ read ($\sigma=0,1$)
$$
 \cL_{\sigma}(\rho)= - i [\bH_\sigma,\rho]
 + \sum_{\nu} \bL_{\sigma,\nu}\rho \bL_{\sigma,\nu}^\dag - \tfrac{1}{2} \Big( \bL_{\sigma,\nu}^\dag \bL_{\sigma,\nu}\rho+
 \rho \bL_{\sigma,\nu}^\dag \bL_{\sigma,\nu} \Big)
$$
with $\bH_\sigma$ Hermitian operator and $\bL_{\sigma,\nu}$ any operator not necessarily Hermitian.

Assume that for $\epsilon=0$ and any initial condition $\rho_0$, $\rho_t$ converges exponentially towards a steady state depending a priori on $\rho_0$. This means that we have a TPCP map  $\ocK$ such that for any $\rho_0$:
\begin{equation}\label{eq:Kchanel}
\lim_{t\mapsto +\infty} e^{t\cL_0}(\rho_0) \triangleq \ocK(\rho_0)
\end{equation}
The range of $\ocK$ is denoted by $\cD_{0}$, the linear space of equilibria for $\cL_0$ corresponding to its kernel. Denote by $\bar d$ the dimension of $\cD_{0}$ and consider an orthonormal basis of $\cD_{0}$ made of $\bar d$ Hermitian operators $\bS_1$, \ldots, $\bS_{\bar d}$ such that $\Tr{\bS_d\bS_{d'}}=\delta_{d,d'}$.
To each $\bS_d$ is associated an invariant operator $\bJ_d=\lim_{t\mapsto +\infty} e^{t \cL_0^*}(\bS_d)$ being a steady-state of the adjoint dynamics (according to the Frobenius Hermitian product) $\dotex \bJ = \cL_0^{*}(\bJ)$ where $\cL_0^*$ is the adjoint of $\cL_0$ (see, e.g., \cite{AlberJ2014PRA}). For any solution $\rho_t$ of~\eqref{eq:dynL0L1} with $\epsilon=0$, $\Tr{\bJ_d\rho_t}$ is constant since $\mathcal{L}_0^*(\bJ_d)=0$ implies that
$$
\dotex \Tr{\bJ_d\rho_t}= \Tr{\bJ_d \mathcal{L}_0(\rho_t)}=  \Tr{ \mathcal{L}_0^*(\bJ_d)\rho_t}=0
.
$$
Thus one has:
\begin{equation}\label{eq:cK0}
\lim_{t\mapsto +\infty}\rho_t = \sum_{d=1}^{\bar d} \Tr{\bJ_d \rho_0} \bS_d \triangleq\ocK(\rho_0).
\end{equation}
Moreover $\Tr{\bJ_d \bS_{d'}} = \delta_{d,d'}$ since for any $t >0$
$$
\Tr{e^{t \cL_0^*}(\bS_d)~ \bS_{d'}}=\Tr{\bS_d~ e^{t \cL_0}(\bS_{d'}) }
$$
and $e^{t \cL_0}(\bS_{d'})= \bS_{d'}$.
Notice that the operator sub-space of co-dimension $\bar d$ defined by
$$
\Big\{ \rho ~\big|~ \forall d\in\{1,\ldots,\bar d\},~\Tr{\bJ_d \rho}=0 \Big\}
$$
corresponds to the set of trajectories exponentially converging to $0$, i.e. the fast invariant sub-space when $\epsilon=0$.

\section{Asymptotic expansion of the slow dynamics} \label{sec:slowManifold}

For $\epsilon >0$ and small, \eqref{eq:dynL0L1} admits also a $\bar d$ dimensional linear invariant subspace denoted by $\cD_{\epsilon}$ and close to $\cD_{0}$ (see~\cite{kato-book-66} for a mathematical justification). This means that the set of real variables $x_1=\Tr{\bJ_1 \rho}$, \ldots, $x_{\bar d}=\Tr{\bJ_{\bar d}\rho}$ can be chosen as local coordinates on $\cD_{\epsilon}$ with the perturbed operator basis $\bS_{1}(\epsilon)$, \ldots $\bS_{\bar d}(\epsilon)$.
If at some time $t$, the solution $\rho_t$~of the perturbed system~\eqref{eq:dynL0L1}, belongs to $\cD_{\epsilon}$, it remains on $\cD_{\epsilon}$ at any time:
$\dotex \rho_t =(\cL_0+\epsilon\cL_1)(\rho_t) $ where $\rho_t= \sum_{d=1}^{\bar d} x_d(t) \bS_{d}(\epsilon)$.
Thus for any $(x_1(t), \ldots, x_{\bar d}(t)) \in\mathbb{R}^{\bar d}$, this invariance property reads
\begin{equation}\label{eq:InvCond}
\sum_{d=1}^{\bar d} \frac{dx_d}{dt} ~ \bS_{d}(\epsilon)
= \left(\cL_0+ \epsilon \cL_1\right)\left(\sum_{d=1}^{\bar d} x_d \bS_{d}(\epsilon)
 \right)
 .
\end{equation}
Thus for any $d\in\{1,\ldots,\bar d\}$, $\frac{dx_d}{dt}$ depends linearly on $x=(x_1,\ldots,x_{\bar d})$, i.e.
$$
\dotex x_d = \sum_{d'}F_{d,d'}(\epsilon) x_{d'}
$$
where  $F_{d,d'}(\epsilon)$ are real coefficients to be  chosen in order to satisfy the
invariance conditions:
$$
\sum_{d'=1}^{\bar d} F_{d',d}(\epsilon) \bS_{d'}(\epsilon) = (\cL_0+\epsilon\cL_1)(\bS_d(\epsilon))
$$
for all  $d \in\{1,\ldots,\bar d\}$.
With the asymptotic expansion
$$
F_{d,d'}(\epsilon) =\sum_{n\geq 0} \epsilon^n F_{d,d'}^{(n)}, \quad \bS_{d}(\epsilon) = \sum_{n\geq 0} \epsilon^n \bS_{d}^{(n)}
$$
one can compute recursively $F_{d,d'}^{(n)}$ and $\bS_{d}^{(n)}$ from $F_{d,d'}^{(m)}$ and $\bS_{d}^{(m)}$ with $m < n$. The recurrence relationships is based on the identification of terms with same orders versus $\epsilon$ in the following equations: $\forall d \in\{1,\ldots,\bar d\}$
$$
  \sum_{d'=1}^{\bar d} \left(\sum_{n\geq 0} \epsilon^n F_{d',d}^{(n)}\right) \left(\sum_{n'\geq 0}\epsilon^{n'} \bS_{d'}^{(n')} \right)
  = \left(\cL_0 +\epsilon\cL_1\right)
\left(\sum_{n\geq 0} \epsilon^n\bS_{d}^{(n)} \right) .
$$
The zero-order term is satisfied with
$$F^{(0)}_{d,d'}=0, \quad \bS_{d}^{(0)} = \bS_d.$$ First-order conditions read: $\forall d \in\{1,\ldots,\bar d\}$
 $$
 \sum_{d''=1}^{\bar d} F_{d'',d}^{(1)} \bS_{d''}^{(0)} = \cL_0(\bS_{d}^{(1)} ) + \cL_1( \bS_{d}^{(0)})
.
 $$
Left multiplication by operator $\bJ_{d'}$ and taking the trace yields
 \begin{equation}\label{eq:F1}
 F_{d',d}^{(1)} = \Tr{\bJ_{d'} \cL_1( \bS_{d}{^{(0)}})}
 \end{equation}
since $\Tr{\bJ_{d'} \bS_{d''}^{(0)}}= \delta_{d',d''}$ and $\Tr{\bJ_{d'}\cL_0(\bW)}=0$ for any operator $\bW$ because $\cL_0^*(\bJ_{d'})=0$.
Thus $\bS_{d}^{(1)}$ is a solution $\bX$ of the following equation:
$$
 \cL_0(\bX) = \sum_{d'} \Tr{\bJ_{d'} \cL_1(\bS^{(0)}_d )} \bS_{d'} - \cL_1(\bS_d^{(0)})
 = \ocK\big(\cL_1(\bS_d^{(0)}) \big) - \cL_1(\bS_d^{(0)})
$$
where the quantum channel $\ocK$ is defined in~\eqref{eq:Kchanel}. Following~\cite{AzouitQST2017}, the general solution $\bX$ is given by the absolutely converging integral,
$$
 \int_{0}^{+\infty} e^{s\cL_0}\left(\cL_1(\bS_d^{(0)}) - \ocK\big(\cL_1(\bS_d^{(0)}) \big) \right) ~ds +\bW
$$
where $\bW$ belongs to $\cD_0$ the kernel of $\cL_0$. We consider here the solution with $\bW=0$ and thus
$$
\bS_{d}^{(1)}= \int_{0}^{+\infty} e^{s\cL_0}\left(\cL_1(\bS_d^{(0)}) - \ocK\big(\cL_1(\bS_d^{(0)}) \big) \right) ~ds
$$
where for all $d'$, $\Tr{\bJ_{d'} \bS_{d}^{(1)}}=0$.
The super-operator $\ocR$ defined for any operator $\bW$ by
\begin{equation}\label{eq:R}
\ocR(\bW) = \int_{0}^{+\infty} e^{s\cL_0}\left(\bW - \ocK\big(\bW\big) \right) ~ds
\end{equation}
provides thus the unique solution $\bX=\ocR(\bW)$ of
$\cL_0(\bX)= \ocK(\bW) - \bW $
such that for all $d$, $\Tr{\bJ_d \bX}=0$. To summarize the above calculation:
\begin{equation}\label{eq:S1}
\bS_{d}^{(1)}= \ocR\big(\cL_1(\bS_d) \big)
.
\end{equation}

Take $n\geq 2$ and assume that we have computed all the terms $F_{d',d}^{(r)}$ and $\bS_{d'}^{(r)}$ of order $r< n$ with
 $\Tr{\bJ_{d'} \bS_{d}^{(r)}}=0$ for all $d$ and $d'$. Invariance conditions of order $n$ read: $\forall d \in\{1,\ldots,\bar d\}$
$$
\sum_{d''=1}^{\bar d} \sum_{r=1}^{n} F_{d'',d}^{(r)}
 \bS_{d''}^{(n-r)} = \cL_0(\bS_{d}^{(n)}) + \cL_1(\bS_{d}^{(n-1)})
 .
$$
Left multiplication by operator $\bJ_{d'}$ and taking the trace yields
\begin{equation}\label{eq:Fn}
 F_{d',d}^{(n)} = \Tr{\bJ_{d'} \cL_1( \bS_{d}^{(n-1)})}
 .
\end{equation}
Then $\bS_{d}^{(n)}$ is as follows
\begin{equation}\label{eq:Sn}
 \bS_{d}^{(n)} = \ocR\left( \cL_1(\bS_{d}^{(n-1)}) - \sum_{d''=1}^{\bar d} \sum_{r=1}^{n} F_{d'',d}^{(r)}
 \bS_{d''}^{(n-r)} \right)
\end{equation}
and satisfies for all $d'$, $\Tr{\bJ_{d'} \bS_{d}^{(n)}}=0$.

 With such asymptotic expansion, we get an order $n$ approximation of the dynamics on the invariant slow-manifold $\cD_\epsilon$.
 \begin{lemma} \label{lem:slowdyn}
Take the slow-fast dynamics~\eqref{eq:dynL0L1}. For $\epsilon$ small enough, it admits a unique invariant slow manifold of dimension $\bar d$ with local real coordinates $(x_1, \ldots, x_{\bar d})$ and  analytic  versus $\epsilon$
$$
\rho = \sum_{d=1}^{\bar d} x_{d} \left(\sum_{n=0}^{+\infty} \epsilon^n \bS_d^{(n)}\right)
$$
Moreover, the evolution on this invariant slow linear subspace is governed by the following linear system, analytic versus $\epsilon$
\begin{equation}\label{eq:dynX}
 \dotex x(t) =\left(\sum_{n=1}^{\infty} \epsilon^n F^{(n)}\right) x(t)
\end{equation}
where $F^{(n)}$ is the matrix of real entries $F^{(n)}_{d,d'}$.
Here $ \bS_d^{(n)}$ and $F^{(n)}_{d,d'}$ are given recursively by~\eqref{eq:Fn} and~\eqref{eq:Sn} starting with $\bS_d^ {(0)}=\bS_d$ and
$F^{(0)}_{d,d'}=0$.
 \end{lemma}
 The precise proof of this lemma is based on the above calculations and on the approximation theorem 5, page 32 of~\cite{carr-book}. Uniqueness and analyticity versus $\epsilon$ is automatically guaranteed since~\eqref{eq:dynL0L1} is linear, time-invariant and of finite dimension.

 %When the infinite absolutely convergent series $\sum_{n=0}^{+\infty} \epsilon^n S_d^{(n)}$ and $\sum_{n=1}^{\infty} \epsilon^n F^{(n)}$ are truncated to some $\bar n$, theorem 5, page 32 of~\cite{carr-book} ensures also that one get an approximation of the invariant manifold and its dynamics up to $\epsilon^{\bar n+1}$ corrections.

 \section{Asymptotic expansion of invariant fast manifold} \label{sec:fastMani}

Section~\ref{sec:slowManifold} was devoted to the invariant slow subspace of~\eqref{eq:dynL0L1}.
To compute the equation of the invariant fast subspace, it is in fact enough to search for its invariant operators equations, i.e., $\bar d$ independent linear scalar equations on $\rho$.
For $\epsilon=0$, the fast subspace corresponds to the solutions $\rho_t$ of $\dotex\rho=\cL_0(\rho_t)$ converging to $0$. They are characterized by the following $\bar d$ linearly independent equations:
$$
\forall d \in\{1,\ldots,\bar d\}, \quad \Tr{\bJ_d \rho}=0
.
$$
Thus for $\epsilon >0$, we are looking for the following set of $\bar d$ equations,
 $$\forall d \in\{1,\ldots,\bar d\}, \quad
 \Tr{\bJ_d(\epsilon) \rho} =0
 $$
 where $\bJ_d(\epsilon)=  \sum_{n\geq 0} \epsilon^n \bJ_d^{(n)}$ with $\bJ_d^{(0)}=\bJ_d$.
 Invariance conditions mean that for all $\rho$ such that $\Tr{\bJ_{d'}(\epsilon)\rho}=0$ for all $d'$, we have,
 $$
\forall d \in\{1,\ldots,\bar d\}, \quad \dotex \Tr{\bJ_d(\epsilon) (\cL_0(\rho)+ \epsilon\cL_1(\rho)}=0
.
 $$
 Thus exists a square matrix of entries $G_{d,d'}(\epsilon)$ depending analytically on $\epsilon$ such that, $\forall d \in\{1,\ldots,\bar d\}$,
 $$
\big( \cL_0^{*}+ \epsilon \cL_1^*\big) (\bJ_d(\epsilon)) = \sum_{d'} G_{d,d'}(\epsilon) \bJ_{d'}(\epsilon)
 $$
 Setting $G_{d,d'}(\epsilon)= \sum_{n\geq 0} \epsilon^n G_{d,d'}^{(n)}$ and identifying
terms of the same power $n$ versus $\epsilon$ give:
\begin{itemize}
 \item for $n=0$, $G_{d,d'}^{(0)} =0$ since $\cL_0^{*}(\bJ_d)=0$.

 \item for $n=1$:
 $
 \cL_0^{*}(\bJ_{d}^{(1)}) + \cL_1^*(\bJ_{d}) = \sum_{d'} G_{d,d'}^{(1)} \bJ_{d'};
 $
 taking the trace with $\bS_{d'}$ gives $ G_{d,d'}^{(1)} = \Tr{\bS_{d'}\cL_1^*(\bJ_{d})}$ and
 $
 \bJ_d^{(1)} = \ocR^*\left( \cL_1^*(\bJ_{d}) - \sum_{d'} G_{d,d'}^{(1)} \bJ_{d'} \right)
 $
 where
 \begin{equation}\label{eq:Rad}
\ocR^*(\bW) = \int_{0}^{+\infty} e^{s\cL_0^*}\left(\bW - \ocK^*\big(\bW\big) \right) ~ds
\end{equation}
and $\ocK^*= \lim_{s\mapsto +\infty}  e^{s\cL_0^*}$.

 \item for $n\geq 2$, we have
 \begin{equation}\label{eq:Gn}
 G_{d',d}^{(n)} = \Tr{\bS_{d'} \cL_1^*( \bJ_{d}^{(n-1)})}
 .
\end{equation}
 and
\begin{equation}\label{eq:Jn}
 \bJ_{d}^{(n)} = \ocR^*\left( \cL_1^*(\bJ_{d}^{(n-1)}) - \sum_{d''=1}^{\bar d} \sum_{r=1}^{n} G_{d'',d}^{(r)}
 \bJ_{d''}^{(n-r)} \right)
 .
\end{equation}
\end{itemize}
We have thus proved the following lemma
\begin{lemma}\label{lem:fastInv}
For the slow/fast system~\eqref{eq:dynL0L1} and $\epsilon$ small enough, exists $\bar d$ linearly independent Hermitian operators $\bJ_d(\epsilon)$, analytic versus $\epsilon$ such that the linear subset of Hermitian operators
$$\Big\{ \rho ~\big| ~\forall d\in\{1,\ldots,\bar d\}, \Tr{\bJ_d(\epsilon) \rho} = 0 \Big\}$$
is invariant. Any trajectory of~\eqref{eq:dynL0L1} starting in this subset converges exponentially to zero with a strictly positive rate independent of $\epsilon$. Moreover by construction,  for all $\forall d,d'\in\{1,\ldots,\bar d\}$, $\Tr{\bJ_d(\epsilon) \bS_{d'}}=\delta_{d,d'}$.
\end{lemma}

\section{Slow propagator and TPCP maps} \label{sec:slowPropa}

Take $T>0$. Then the linear map $\cK_{\epsilon,T}$ on $\cD_0$:
$$
\cD_0 \ni \rho_0 \mapsto \cK_{\epsilon,T}(\rho_0) \triangleq \ocK\Big( e^{T(\cL_0 +\epsilon \cL_1)}(\rho_0) \Big)\in\cD_0
$$
is TPCP.   Take $\epsilon$ small enough and consider the $\bar d\times \bar d$ matrix $E(\epsilon)$ of entries
 $$
 E_{d,d'}(\epsilon) = \Tr{\bJ_d(\epsilon) \bS_{d'}(\epsilon)}
 $$
 where $\bS_{d'}(\epsilon)=\sum_{n\geq 0} \epsilon^n \bS_{d'}^{(n)}$ with $\bS_{d'}=\bS_{d'}^{(0)}$ and $\bJ_d(\epsilon)=\sum_{n\geq 0} \epsilon^n \bJ_{d}^{(n)}$ with $\bJ_{d}=\bJ_{d}^{(0)}$ are defined in lemma~\ref{lem:slowdyn} and lemma~\ref{lem:fastInv} respectively.
 Set $\rho_0 = \sum_d x_d(0) \bS_d$ with $x(0)=(x_1(0),\ldots,x_{\bar d}(0)) \in\mathbb{R}^{\bar d} $ and consider
 the solution
$z_{\epsilon}(0)\in\mathbb{R}^{\bar d}$ of the linear system
$$
E(\epsilon) z_\epsilon(0) = x(0)
.
$$
From $\Tr{\bJ_d \bS_{d'}(\epsilon)}=\Tr{\bJ_d(\epsilon) \bS_{d'}}=\delta_{d,d'}$ one has
$$
E(\epsilon) = I_{\bar d} +\epsilon^2 C(\epsilon)
$$
where $I_{\bar d}$ is the identity matrix of size $\bar d$ and the matrix $C(\epsilon)$ is analytic versus $\epsilon$. Thus $E(\epsilon)$ is invertible for $\epsilon$ small with
$E^{-1}(\epsilon) = I_{\bar d} + O(\epsilon^2)$.
%Set $\mathcal{G}_{\epsilon,T} = e^{ T F(\epsilon) }$ where the $\bar d\times \bar d$ matrix $F(\epsilon)= \sum_{n=1}^{\infty} \epsilon^n F^{(n)}$ is given in lemma~\ref{lem:slowdyn}.
Then we have the following lemma
\begin{lemma}\label{lem:Approx}
 For $\epsilon$ small enough, exist $\gamma >0$ and $M>0$ such that for any $T >0$ and any $x(0)\in\mathbb{R}^{\bar d}$ we have
 $$
 \sum_{d}\left| \Tr{\bS_{d} \cK_{\epsilon,T}(\rho_0) } - z_{\epsilon,d}(T) \right| \leq M e^{-\gamma T} \sqrt{\Tr{\rho^2_0}}
 $$
 where $z_{\epsilon}(T)= e^{ T F(\epsilon) } E^{-1}(\epsilon) x(0) $ and $\rho_0= \sum_d x_d(0) \bS_d$.
\end{lemma}
\begin{proof}
It is based on the following arguments inspired by~\cite[theorem 2, point b, page 4]{carr-book} and by the notion of shadow trajectories around hyperbolic invariant manifold~(see, e.g.,\cite{Hirsh_shadowing_1993} and~\footnote{\scriptsize{\tt
https://spartacus-idh.com/liseuse/116/\symbol{35}page/142}}).
Since $\ocK^*(\bS_d)=\bJ_d$ one has
$$
 \Tr{\bS_{d} \cK_{\epsilon,T}(\rho_0) } = \Tr{\bJ_d e^{T(\cL_0+\epsilon\cL_1)}(\rho_0)}
$$
The initial state on the invariant manifold $ \rho_{\epsilon,0}=\sum_d z_{\epsilon,d}(0) \bS_d(\epsilon)$ and $\rho_0$ are such that
$\rho_0-\rho_{\epsilon,0}$ belongs to the invariant fast manifold. Thus exist $\gamma >0$ and $M>0$ such that
$$
\left\|e^{T(\cL_0+\epsilon \cL_1)} (\rho_0-\rho_{\epsilon,0})\right\| \leq M e^{- \gamma T} \left\|\rho_0-\rho_{\epsilon,0}\right\|
.
$$
We conclude with $e^{T(\cL_0+\epsilon \cL_1)} (\rho_{\epsilon,0})= \sum_d z_{\epsilon,d}(T) \bS_d(\epsilon)$ and
$\Tr{\bJ_d \bS_{d'}(\epsilon)} = \delta_{d,d'}$.
\end{proof}

Take $\bar T>0$. This lemma implies that for $T = \bar T /\epsilon$ corresponding to a slow time-scale,  the approximation of $\cK_{\epsilon,\frac{\bar T}{\epsilon }}(\rho_0)$ by
$\sum_d z_{\epsilon,d}(\frac{\bar T}{\epsilon }) \bS_d$ is exponentially precise. Thus, if we restrict the asymptotic expansion to second-order terms and consider evolution over a slow time-scale $\frac{\bar T}{\epsilon }$, we are close up to order $2$ terms to a TPCP map as stated  in the following lemma.
\begin{lemma} \label{lem:approx2}
Consider the slow/fast system~\eqref{eq:dynL0L1}, the second-order approximation of the slow dynamics
$$
\dotex x = \big(\epsilon F^{(1)} + \epsilon^2 F^{(2)}\big) x
$$
and its propagator matrix $\mathcal{G}^{(2)}_{\frac{\bar T}{\epsilon}}= e^{\frac{\bar T}{\epsilon }(\epsilon F^{(1)} + \epsilon^2 F^{(2)})}$ over the slow time scale $\bar{T}/\epsilon$.
For any $\bar T >0$, exists $M_{\bar T} >0$ such that for any $x(0)\in\mathbb{R}^{\bar d}$ we have
$$
\left\| \cK_{\epsilon,\tfrac{\bar T}{\epsilon }}(\rho_0) - \rho(\tfrac{\bar T}{\epsilon }) \right\| \leq M_{\bar T} \epsilon^2 \|x(0)\|
$$
where $\rho_0=\sum_d x_d(0) \bS_d$ and $\rho(\tfrac{\bar T}{\epsilon })=\sum_{d} x_d(\tfrac{\bar T}{\epsilon }) \bS_d$
with $x(\tfrac{\bar T}{\epsilon })= \mathcal{G}^{(2)}_{\frac{\bar T}{\epsilon}} x(0)$.
\end{lemma}
Notice that
$$F_{d',d}^{(1)} = \Tr{\bJ_{d'} \cL_1(\bS_{d})}$$ and ( $\ocR$ defined in~\eqref{eq:R})
 $$ F^{(2)}_{d',d}
 = \Tr{ \cL_1^{*} (\bJ_{d'}) ~ \ocR\Big( \cL_1(\bS_{d})\Big) }
$$
can be computed from the nominal operators $\bS_d$ and $\bJ_{d'}$ defined in section~\ref{sec:slowfastdyn}.
\begin{proof}
It combines lemma~\ref{lem:Approx}, the fact that $\mathcal{G}^{(2)}_{\frac{\bar T}{\epsilon}}= e^{\frac{\bar T}{\epsilon } F(\epsilon) } + O(\epsilon^2)$ since $\frac{\bar T}{\epsilon } F(\epsilon) = \bar T (F^{(1)} + \epsilon F^{(2)}) + O(\epsilon^2) $ and that $ E^{-1}(\epsilon) x(0) = x(0)+ O(\epsilon^2)$ since $E^{-1}(\epsilon)= I_{\bar d} + O(\epsilon^2)$.
\end{proof}

%%%%%%%%%%%%%%%%%%%%%%%%%%%%%%%%%%%%%%%%%%%%%%%%%%%%%%%%%%%%%%%%%%
\section{Concluding remarks}\label{sec:conclusion}

For infinite dimensional systems, mathematical justifications of such asymptotic expansions  are not straightforward (existence, uniqueness  and convergence of the series). For infinite dimensional systems constructed with mainly bounded operators, precise mathematical guaranties are available (see, e.g., \cite{AULBACH200091,HOCHS20197263}). One cannot use these available results when the dynamics rely on unbounded operators parameterizing the GKSL dynamical model~\eqref{eq:dynL0L1}.
 The mathematical justification of such infinite dimensional extension requires precise functional analysis investigations and assumptions. They will be  addressed in future developments.

These asymptotic expansion can be exploited numerically to simulate on a classical computer,  composite open quantum systems encountered in quantum error correction (see, e.g., \cite{regent2023adiabatic} for preliminary results with cat-qubit systems).

%%%%%%%%%%%%%%%%%%%%%%%%%%%%%%%%%%%%%%%%%%%%%%%%%%%%%%%%%%%%%%%%%%%%%%%%%%%%%%%
\section*{Acknowledgments}

We thank Philippe Campagne-Ibarcq, J\'{e}r\'{e}mie Guillaud, Mazyar Mirrahimi, Claude Le Bris, Alain Sarlette, Lev-Arcady Sellem and Antoine Tilloy  for numerous discussions and scientific exchanges on model reduction, numerical simulations, cat-qubits and bosonic codes.
We thank also Laurent Praly for precious indications on the concept of shadow trajectories around hyperbolic invariant submanifold.

This project has received funding from the Plan France 2030 through the project ANR-22-PETQ-0006.

This project has received funding from the European Research Council (ERC) under the European Union’s Horizon 2020 research and innovation programme (grant agreement No. [884762]).

\bibliographystyle{unsrt} %\bibliographystyle{plain}
%\bibliography{C:/PRData/Jabref_PR/RouchonJabref}

\end{document}